\documentclass[11pt]{article}
\usepackage{amsmath}
\usepackage{amssymb}
\usepackage{gastex}
\usepackage{amsthm}
\usepackage{cite}

\DeclareSymbolFont{rsfscript}{OMS}{rsfs}{m}{n}
\DeclareSymbolFontAlphabet{\mathrsfs}{rsfscript}

\theoremstyle{plain}
\newtheorem{theorem}{Theorem}[section]
\newtheorem{lemma}[theorem]{Lemma}

\newcommand{\sa}{synchronizing automata}
\newcommand{\sza}{synchronizing $0$-au\-to\-mata}
\newcommand{\san}{synchronizing automaton}
\newcommand{\szan}{synchronizing $0$-au\-to\-ma\-ton}

\theoremstyle{remark}

\begin{document}
\title{A new lower bound for reset threshold\\
of synchronizing automata with sink state}
\author{D. S. Ananichev\thanks{The author acknowledges support by the Russian Foundation for Basic Research, grant no.\ 16-01-00795,
and the Competitiveness Program of Ural Federal University.}}
\date{Institute of Mathematics and Computer Science\\
Ural Federal University\\ 620000 Ekaterinburg, RUSSIA\\
{\small\textsl{D.S.Ananichev@urfu.ru}}}
\maketitle

\thispagestyle{empty}

\begin{abstract}
We present a new series of examples of binary slowly synchronizing automata with sink state.
The reset threshold of the $n$-state automaton in this series is $\frac{n^2}{4}+2n-9$. 
This improves on the previously known lower bound for the maximum reset threshold
of binary synchronizing $n$-state automata with sink state.
\end{abstract}

\section{Background and motivation}

Let $\mathrsfs{A}=\langle Q,\Sigma,\delta\rangle$ be a deterministic finite automaton (DFA, for short) with the state set $Q$, the input alphabet $\Sigma$, and the transition function $\delta:Q\times \Sigma\to Q$. If $|\Sigma|=2$ then we refer to this automaton as a \emph{binary} DFA. The action of the letters in $\Sigma$ on the states in $Q$ defined via $\delta$ extends in a natural way to an action of the words in the free $\Sigma$-generated monoid $\Sigma^*$; the latter action is still denoted by $\delta$.  For any $w\in\Sigma^*$ and $X\subseteq Q$, we set $\delta(X,w)=\{\delta(q,w)\mid q\in X\}$. Sometimes we write $X.w$ for $\delta(X,w)$.

A DFA $\mathrsfs{A}=\langle Q,\Sigma,\delta\rangle$ is said to be \emph{synchronizing} if there is a word $w\in\Sigma^*$ such that $|\delta(Q,w)|=1$. The word $w$ is then called a \emph{synchronizing} or \emph{reset} word for $\mathrsfs{A}$. The minimum length of reset words for a \san\ $\mathrsfs{A}$ is called the \emph{reset threshold} of $\mathrsfs{A}$ and is denoted by $rt(\mathrsfs{A})$. The reset threshold of a class $\mathcal{C}$ of \sa\ is defined as $rt(\mathcal{C}):=\max\{rt(\mathrsfs{A})\mid \mathrsfs{A}\in\mathcal{C}\}$.

\v{C}ern\'{y}~\cite{Ce64} constructed for each positive integer $n$ an $n$-state binary \san\ with reset threshold $(n-1)^2$. The famous
\emph{\v{C}ern\'{y} conjecture} claims the optimality of this construction, that is, $(n-1)^2$ is conjectured to be the precise value for the reset threshold for the class \sa\ with $n$ states. The conjecture remains open for more than 50 years and is arguably the most longstanding open problem in the combinatorial theory of finite automata.

Upper bounds within the confines of the \v{C}ern\'{y} conjecture have been obtained for the reset thresholds of some special classes
of \sa, see, e.g.,~\cite{Ep90,Ry97,Du98,Ka03,AV04a,AV04b,AV05,Tr-JALC}. One of these classes is the class of automata with sink state. A state $z$ of a DFA $\mathrsfs{A}=\langle Q,\Sigma,\delta\rangle$ is said to be a \emph{sink state} (or \emph{zero}) if $\delta(z,a)=z$ for all $a\in\Sigma$. It is clear that a \san\ may have at most one sink state and each word that resets a \san\ possessing sink state must bring all states to sink state. We refer to \sa\ with sink state as \emph{\sza}.

A rather straightforward argument shows that every $n$-state \szan\ can be reset by a word of length $\frac{n(n-1)}{2}$, see, e.g.,~\cite{Ry97}. This upper bound is in fact tight because, for each $n$, there exists a \szan\ with $n$ states and $n-1$ input letters which cannot be reset by any word of length less than $\frac{n(n-1)}{2}$. Such an automaton\footnote{We were not able to trace the origin of this series of \sa. It is contained, for instance, in~\cite{Ry97} but it should have been known long before~\cite{Ry97} since a very similar series had appeared already in~\cite{KW71}.} is shown in Fig.~\ref{oldexample} where $\Sigma:=\{a_1,\dots,a_{n-1}\}$.
\begin{figure}[th]
\begin{center}
\unitlength=1mm
\begin{picture}(134,25)(3,0)
\node(A)(10,10){0} \node(B)(32,10){1} \node(C)(54,10){2}
\node(D)(76,10){$3$}
\multiput(85,10)(3,0){3}{\circle*{.75}}
\node(E)(100,10){$n{-}2$}\node(F)(122,10){$n{-}1$}
\drawedge(B,A){$a_1$} \drawedge(C,B){$a_2$}\drawedge(B,C){}
\drawedge(D,C){$a_3$}\drawedge(C,D){}
\drawedge(F,E){$a_{n-1}$}\drawedge(E,F){}
\drawloop[ELpos=50,loopangle=90](A){\footnotesize $\Sigma$}
\drawloop[ELpos=50,loopangle=90](B){\footnotesize $\Sigma \!\setminus\! \{a_1,a_2\}$}
\drawloop[ELpos=50,loopangle=90](C){\footnotesize $\Sigma \!\setminus\! \{a_2,a_3\}$}
\drawloop[ELpos=50,loopangle=90](D){\footnotesize $\Sigma \!\setminus\! \{a_3,a_4\}$}
\drawloop[ELpos=50,loopangle=90](E){\footnotesize $\Sigma \!\setminus\! \{a_{n{-}2},a_{n{-}1}\}$}
\drawloop[ELpos=50,loopangle=90](F){\footnotesize $\Sigma \!\setminus\! \{a_{n{-}1}\}$}
\end{picture}
\caption{A 0-automaton whose reset threshold is $\frac{n(n-1)}{2}$}
\label{oldexample}
\end{center}
\end{figure}
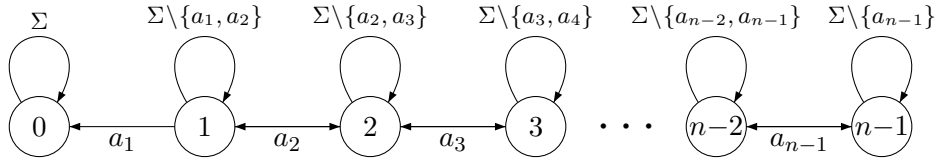

An essential feature of the example in Fig.~\ref{oldexample} is that the input alphabet size grows with the number of states. This contrasts with the aforementioned examples due to \v{C}ern\'{y}~\cite{Ce64} in which the alphabet is independent of the state number
and leads to the following natural problem: \emph{to determine the reset threshold of $n$-state \sza\ over a fixed input alphabet}.
To the best of our knowledge, this problem remains open. It appears to be of independent interest and has some connection with some questions of formal language theory related to so-called complete sets of words, see~\cite{Pri11}. 

Up to now, the best lower bound for the reset threshold of binary \sza\ with $n$ states that holds for arbitrarily large $n$  has been found by Martyugin~\cite{Ma08}. Namely, he has constructed, for every $n\geq 8$, a binary \szan\ $\mathrsfs{M}_n$ with $n$ states such that $rt(\mathrsfs{M}_n)=\left\lceil\frac{n^2+6n-16}4\right\rceil$. Besides that, Martyugin provided an isolated example of a 10-state \szan\ with reset threshold 37, thus exceeding $\frac{10^2+6\cdot 10-16}4=36$. Vorel in his thesis~\cite{Vo15} has extended Martyugin's example to a series of \sza\ $\mathrsfs{V}_j$ with $n=12j-2$ of states and conjectured \cite[Conjecture 2.13]{Vo15} that $rt(\mathrsfs{V}_j)=\frac14n^2+\frac32n-3$. Observe that for even $n$, one has $\left\lceil\frac{n^2+6n-16}4\right\rceil=\frac14n^2+\frac32n-4$ so that the validity of Vorel's conjecture would improve Martyugin's bound just by 1. The conjecture was computationally verified in ~\cite{Vo15} for $j=1,2,3,4,5,6$, the case $j=1$ corresponding to Martyugin's 10-state example. 

In this paper, we use a neat idea from~\cite{Vo15} to provide a new series of binary \sza\ with $n\ge16$ states, $n\equiv4\pmod{12}$. The reset threshold of the $n$-th DFA in our series equals $\frac14n^2+2n-9$, thus improving on both the bound established in~\cite{Ma08} and the one conjectured in~\cite{Vo15}.

\section{Appending tails to almost permutation automata}

We reproduce here the description of Martyugin's series of examples which provides the lower bound $\left\lceil\frac{n^2+6n-16}4\right\rceil$ for the reset threshold of binary \sza\ with $n$ states. We restrict ourselves to the case of even $n=2m\geq 8$; this is sufficient to explain our approach. 

Let $\mathrsfs{M}_{2m}$ be the DFA $(Q,\{a,b\},\delta)$, where $Q=\{0,\ldots,2m-1\}$ and the transition function $\delta$ is defined as follows:
\[
\delta (i,a)= \left\{\begin{aligned}
 0 &\ \text{ if }\ {i=0}, \\
 {i-1} &\ \text{ if }\ {i=1,\ldots,m-1}, \\
 {2m-2} &\ \text{ if }\ {i=m}, \\
 {i-1} &\ \text{ if }\ {i=m+1,\ldots,2m-2}, \\
 {2m-1} &\ \text{ if }\ {i=2m-1}; \\
\end{aligned} \right.
\]
\[
\delta (i,b)=\left\{\begin{aligned}
 0 &\ \text{ if }\ {i=0}, \\
 m &\ \text{ if }\ {i=1,\ldots, m-1}, \\
 {m-1} &\ \text{ if }\ {i=m}, \\
 {2m-1} &\ \text{ if }\ {i=m+1}, \\
 {i+1} &\ \text{ if }\ {i=m+2,\ldots, 2m-3}, \\
 {m+1} &\ \text{ if }\ {i=2m-2}, \\
 {m+2} &\ \text{ if }\ {i=2m-1}. \\
\end{aligned} \right.
\]
The automaton is shown in Fig.~\ref{A-even}. Clearly, 0 is the sink state of $\mathrsfs{M}_{2m}$.

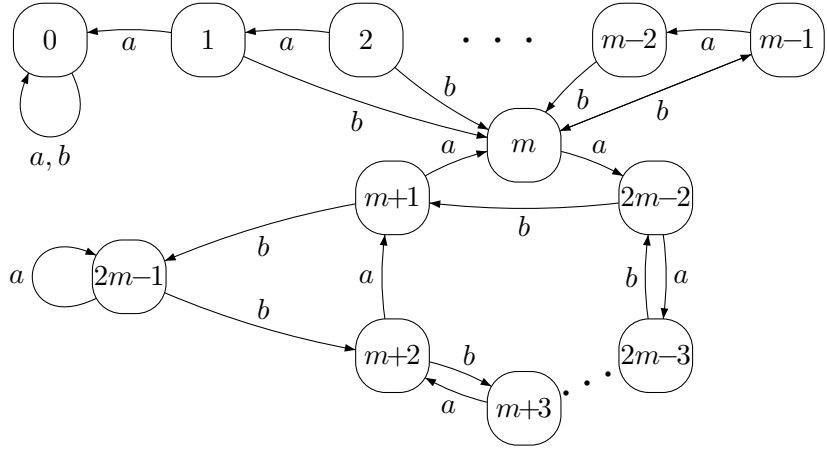
\begin{figure}[th]
\begin{center}
\unitlength 0.7mm
\begin{picture}(180,85)(0,-5)
\gasset{Nw=14,Nh=14} \node(A0)(20,70){0} \node(A1)(50,70){1}
\node(A2)(80,70){2} \multiput(99,70)(6,0){3}{\circle*{1}}
\node(A3)(130,70){$m\!\!-\!\!2$} \node(A4)(160,70){$m\!\!-\!\!1$}
\node(B0)(35,25){$2m\!\!-\!\!1$} \node(B1)(135,40){$2m\!\!-\!\!2$}
\node(B2)(135,10){$2m\!\!-\!\!3$}
\multiput(118,3)(4,1.8){3}{\circle*{1}}
\node(B3)(110,0){$m\!\!+\!\!3$} \node(B4)(85,10){$m\!\!+\!\!2$}
\node(B5)(85,40){$m\!\!+\!\!1$} \node(B6)(110,50){$m$}
\drawloop[ELpos=50,loopangle=-90](A0){\textbf{$a,b$}}
\drawloop[ELpos=50,loopangle=-180](B0){\textbf{$a$}}
\drawedge(A4,B6){\textbf{$b$}}
\drawedge(B6,A4){}
\gasset{curvedepth=-2}
\drawedge(A1,A0){\textbf{$a$}}\drawedge(A2,A1){\textbf{$a$}}\drawedge(A4,A3){\textbf{$a$}}
\gasset{curvedepth=2}
\drawedge(B1,B2){\textbf{$a$}}\drawedge(B3,B4){\textbf{$a$}}
\drawedge(B4,B5){\textbf{$a$}}\drawedge(B5,B6){\textbf{$a$}}
\drawedge(B6,B1){\textbf{$a$}}
\thicklines \gasset{curvedepth=-2}
\drawedge[ELside=r](A1,B6){\textbf{$b$}}\drawedge(A2,B6){\textbf{$b$}}
\drawedge(B5,B0){\textbf{$b$}}\drawedge(B0,B4){\textbf{$b$}}
\drawedge(A3,B6){\textbf{$b$}}
\gasset{curvedepth=2}
\drawedge(B4,B3){\textbf{$b$}}\drawedge(B2,B1){\textbf{$b$}}
\drawedge(B1,B5){\textbf{$b$}}
\end{picture}
\caption{The automaton $\mathrsfs{M}_{2m}$}
\label{A-even}
\end{center}
\end{figure}

It is easy to see that the automaton $\mathrsfs{M}_{2m}$ consists of two parts: the ``body'' formed by the states in
$\{m,m+1,\dots ,2m-1\}$ and the ``tail'' formed by the states in $\{0,1,\dots ,m-1\}$.

Now we describe Vorel's concept of appending a {\it tail\/} of length $k$ to an automaton with sink state. Let $\mathrsfs{A}=\langle Q,\{a,b\},\delta\rangle$ be a binary DFA with sink state $q_0\in Q$. Then for each $k\geq 0$ and each $r\in Q$, the expression
$\mathrsfs{A}(k,r)$ stands for the following automaton $\langle Q',\{a,b\},\delta'\rangle$ with $k$ additional states:
\[
Q'=Q\cup\{t_0,t_1,\dots ,t_{k-1}\},
\]

\begin{gather*}
\delta' (s,a)= \left\{\begin{aligned}
\delta (s,a)  &\ \text{ if }\ {s\in Q\setminus\{q_0\}}, \\
 {t_{k-1}} &\ \text{ if }\ {s=q_0}, \\
 {t_0} &\ \text{ if }\ {s=t_0}, \\
 {t_{i-1}} &\ \text{ if }\ {s=t_i,\ i\in\{1,2,\ldots,k-1\}}, 
\end{aligned} \right.\\[1ex]
\delta (s,b)=\left\{\begin{aligned}
\delta (s,b) &\ \text{ if }\ {s\in Q\setminus\{q_0\}}, \\
 {t_0} &\ \text{ if }\ {s=t_0}, \\
 {r} &\ \text{ otherwise, }
\end{aligned} \right.
\end{gather*}
for each $s\in Q'$.  Observe that $t_0$ is a unique sink state of $\mathrsfs{A}(k,r)$.

It is easy to realize that Martyugin's automaton $\mathrsfs{M}_{2m}$ in Fig.~\ref{A-even} fits under the framework of the above construction; namely, it coincides with the automaton $\mathrsfs{A}_{2m}(m-1,m)$, where $\mathrsfs{A}_{2m}$ is the automaton in Fig.~\ref{A-even2}.

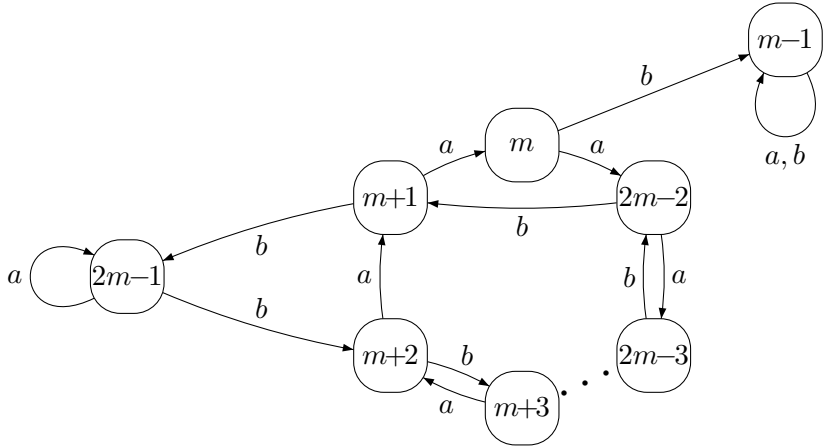
\begin{figure}[th]
\begin{center}
\unitlength 0.7mm
\begin{picture}(180,85)(0,-5)
\gasset{Nw=14,Nh=14}
\node(A4)(160,70){$m\!\!-\!\!1$}
\node(B0)(35,25){$2m\!\!-\!\!1$}
\node(B1)(135,40){$2m\!\!-\!\!2$}
\node(B2)(135,10){$2m\!\!-\!\!3$}
\multiput(118,3)(4,1.8){3}{\circle*{1}}
\node(B3)(110,0){$m\!\!+\!\!3$} \node(B4)(85,10){$m\!\!+\!\!2$}
\node(B5)(85,40){$m\!\!+\!\!1$} \node(B6)(110,50){$m$}
\drawloop[ELpos=50,loopangle=-90](A4){\textbf{$a,b$}}
\drawloop[ELpos=50,loopangle=-180](B0){\textbf{$a$}}
\drawedge(B6,A4){\textbf{$b$}}
\gasset{curvedepth=-2}
\gasset{curvedepth=2}
\drawedge(B1,B2){\textbf{$a$}}\drawedge(B3,B4){\textbf{$a$}}
\drawedge(B4,B5){\textbf{$a$}}\drawedge(B5,B6){\textbf{$a$}}
\drawedge(B6,B1){\textbf{$a$}}
\thicklines \gasset{curvedepth=-2}
\drawedge(B5,B0){\textbf{$b$}}\drawedge(B0,B4){\textbf{$b$}}
\gasset{curvedepth=2}
\drawedge(B4,B3){\textbf{$b$}}\drawedge(B2,B1){\textbf{$b$}}
\drawedge(B1,B5){\textbf{$b$}}
\end{picture}
\caption{The automaton $\mathrsfs{A}_{2m}$}
\label{A-even2}
\end{center}
\end{figure}

We call a synchronizing binary DFA $(Q,\{a,b\},\delta)$ with sink state $q_0\in Q$ an {\it almost permutation automaton\/} if it fulfils the following three conditions:\\
1. There is a unique state $r\in Q\setminus\{q_0\}$ such that  $\delta(r,b)=q_0$; we refer to $r$ as a \emph{pre-sink} state.\\
2. The letter $b$ acts as a permutation on the set $Q\setminus\{r\}$.\\
3. The letter $a$ acts as a permutation on $Q$.\\
We use the expression $\mathrsfs{A}_{(q_0,r)}$ to denote an almost permutation automaton with sink state $q_0$ and pre-sink state $r$.
We call the least $k$ such that $a^k$ acts as the identity permutation the \emph{order} of $a$. Clearly, the order of $a$ is the least common multiple of the lengths of cycles with respect to $a$.

Observe that the automaton $\mathrsfs{A}_{2m}$ in Fig.~\ref{A-even2} is an almost permutation automaton.

Now we reproduce Vorel's lemma \cite[Lemma 2.14]{Vo15} about adding a tail to an almost permutation automaton. We have included a proof of the lemma for the sake of completeness.

\begin{lemma}
Let $\mathrsfs{A}_{(q_0,r)}=\langle Q,\{a,b\},\delta\rangle$ be an $n$-state synchronizing almost permutation automaton and let $k$ be a multiple of the order of $a$. Then $$rt(\mathrsfs{A}_{(q_0,r)}(k,r))=rt(\mathrsfs{A}_{(q_0,r)})+nk.$$
\end{lemma}

\begin{proof}
First, we show that $rt(\mathrsfs{A}_{(q_0,r)}(k,r))\leq rt(\mathrsfs{A}_{(q_0,r)})+nk$. Let $w$ be a reset word of $\mathrsfs{A}_{(q_0,r)}$. For each $d\geq 1$ we denote by $u_d$ the shortest prefix $u$ of $w$ satisfying $|\delta^{-1}(\{q_0\},u)|\geq d$, which means that $u$ maps at least $d$ states to $q_0$. Clearly, $u_1$ is the empty word and $u_n=w$. Denote $w=v_2v_3\dots v_n$,
where $u_d=u_{d-1}v_d$ for each $d\in\{2,3,\dots ,n\}$. Let $w'=a^kv_2a^kv_3a^k\dots v_na^k$. It is enough to show that $w'$ resets
$\mathrsfs{A}_{(q_0,r)}(k,r)$, i.e., $\delta'(s,w')=t_0$ for each $s\in Q'$. If $s\in\{t_0,t_1,\dots ,t_{k-1}\}$, we just observe that $\delta'(s,a^k)=t_0$. If $s\in Q$, let $d$ be the least integer such that $\delta(s,u_d)=q_0$. Since $\delta(s,u')\ne q_0$ for each proper prefix $u'$ of $u_d$, the definition of $\delta'$ implies that $\delta'(s,u_d)=\delta(s,u_d)=q_0$. Since $k$ is a multiple of the order of $a$, we have $\delta'(s,u_i)=\delta'(s,a^kv_2a^kv_3a^k\dots v_i)$ for every $i\in\{1,\dots ,d\}$. As $\delta'(q_0,a^k)=t_0$, we are done.

Second, we show that $rt(\mathrsfs{A}_{(q_0,r)}(k,r))\geq rt(\mathrsfs{A}_{(q_0,r)})+nk$. Let $w$ be a reset word of $\mathrsfs{A}_{(q_0,r)}(k,r)$. For each $s\in Q$ we denote by $u_s'$ the shortest prefix $u'$ of $w'$ with $\delta'(s,u')=t_0$. Observe that $u_s'=v_s'a^k$ for some $v_s'$ with $\delta'(s,v_s')=q_0$. Moreover, since $\delta'(Q',a)=Q'\setminus \{q_0\}$, each $v_s'$ either ends with $b$ or is empty.

Next, we show that $v_p'\ne v_s'$ for each distinct $p,s\in Q$. Otherwise, we have $\delta'(p,v')=\delta'(s,v')=q_0$, where $v'=v_p'=v_s'$. As $r$ is the only merging state in $\mathrsfs{A}_{(q_0,r)}(k,r)$ except for $t_0$, we have $\delta'(p,v''b)=\delta'(s,v''b)=r$, for some prefix $v''$ of $v'$ with $\delta'(p,v'')\ne \delta'(s,v'')$. Since the states $\delta'(p,v'')$ and $\delta'(s,v'')$ belong to $\delta'^{-1}(r,b)=\{t_0,t_1,\dots ,t_{k-1},q_0\}$, we can denote $t_i=\delta'(p,v'')$ and $t_j=\delta'(s,v'')$,
where $i,j\in \{0,1,\dots ,k\}$ and $t_k$ stands for $q_0$. From $t_i=\delta'(p,v'')$ it follows that $v''$ ends with $ba^i$ or equals to $a^i$. From $t_j=\delta'(s,v'')$ it follows that $v''$ ends with $ba^j$ or equals to $a^j$. As $i\ne j$, we get a contradiction.

Therefore, each of the distinct prefixes $v_s'$ of $w'$ for $s\in Q$ ends with $b$ or is empty and is followed by $a^k$. Thus, $w'$ contains at least $n$ disjoint occurrences of the factor $a^k$. Let $w$ be obtained from $w'$ by deleting these factors, so that we have $|w|\leq |w'|-nk$. It remains to show that $w$ is a reset word of $\mathrsfs{A}_{(q_0,r)}$. Choose $s\in Q$ and let $w_s'$ be the shortest prefix $u'$ of $w'$ with $\delta'(s,u')=q_0$. Since $\delta'(Q',a)=Q'\setminus \{q_0\}$, the word $w_s'$ ends with $b$ or is empty. Thus, we can consider the prefix $w_s$ of $w$ obtained by deleting the occurrences of $a^k$ from $w_s'$. Since $k$ is a multiple of the order of $a$, the word $a^k$ acts as the identity permutation on the set $Q\setminus \{q_0\}$ in both $\mathrsfs{A}_{(q_0,r)}(k,r)$ and $\mathrsfs{A}_{(q_0,r)}$. We conclude that
$$\delta(s,w_s)=\delta(s,w_s')=\delta'(s,w_s')=q_0,$$
which implies easily that $\delta(s,w)=q_0$.
\end{proof}

Thus, \emph{in order to obtain a series of binary $0$-automata with high reset threshold, it is sufficient to construct a series of almost permutation automata with high reset threshold.}

Let us make the idea just stated more precise. Suppose we have constructed a series of almost permutation automata $\mathrsfs{A}_n$ such that 
$$rt(\mathrsfs{A}_n)=A n^2+ B n+ C,$$ 
where $A$, $B$ and $C$ are some constants and $n$ is the number of states. (Of course, $0\leq A\leq 1/2$ in view of the inequality $rt(\mathrsfs{A}_n)\le\frac{n(n-1)}2$ from~\cite{Ry97}.) Then we can add tails of lengths $k=k(n)$ and obtain the series of binary $0$-automata $\mathrsfs{B}_N$ with $N=n+k$ states. If $k(n)$ is chosen to be the order of the letter $a$ in $\mathrsfs{A}_n$, then Vorel's lemma implies that $$rt(\mathrsfs{B}_N)=A n^2+ B n+ C+nk.$$

Suppose that $k=Dn+E$, where $D$ and $E$ are constants. Then,
$$rt(\mathrsfs{B}_N)=\frac{A+D}{(1+D)^2}\cdot N^2+O(N).$$
If $D=1-2A$, then the first coefficient is maximal and is equal to $\frac{1}{4(1-A)}$. This implies that if $A<1/2$, then
$rt(\mathrsfs{B}_N)$ grows faster than $rt(\mathrsfs{A}_n)$.

If we try to apply the above reasoning to Martyugin's example, we may observe that the reset threshold of the almost permutation automata
$\mathrsfs{A}_{2m}$ shown in Fig.~\ref{A-even2} grows linearly with the number of states, that is, we have to look at the special case when $A=0$ (and hence $D=1$). If $rt(\mathrsfs{A}_n)=B n+ C$ and $k=n+E$, then
$$rt(\mathrsfs{B}_N)=\frac{1}{4}\cdot N^2+\frac{B}{2}\cdot N+O(1).$$
In Martyugin's example $B=3$. Therefore, to obtain a larger lower bound for the reset threshold of binary \sza, it is sufficient to
construct a series of almost permutation automata $\mathrsfs{A}_n$ with constant order of the letter $a$ and such that $rt(\mathrsfs{A}_n)=A n^2+ B n+ C$, where $A>0$, or $A=0$ and $B>3$. We present such a construction in the next section.

\section{A series of slowly syncronizable\protect\\ almost permutation automata}

We present a series of almost permutation automata $\mathrsfs{A}_n=\langle Q,\Sigma,\delta\rangle$, where $n=7,8,\dots$ is the number
of states, such that the reset threshold for the automaton $\mathrsfs{A}_n$ is at least $4n-13$. The state set $Q_n$ of the automaton
$\mathrsfs{A}_n$ is $\{0,1,2,\dots,n-1\}$. The input alphabet $\Sigma$ of $\mathrsfs{A}_n$ consists of two letters $a$ and $b$.

The actions of the letters $a$ and $b$ on the set $Q_n$ are defined as follows:
$$
\begin{array}{l}
\delta(0,a)=\delta(0,b)=0,\\
\delta(1,a)=0,\\
\delta(1,b)=3,\ \delta(2,b)=1,\ \delta(3,b)=2,\\
\delta(2,a)=4,\ \delta(3,a)=2,\ \delta(4,a)=3,\\
\mbox{ and for the other states}\\
\delta(j,a)=
\begin{cases}
j-1 &\text{if}\ j \text{ is odd},\\
j+1 &\text{if}\ j \text{ is even and } j\not=n-1,\\
j   &\text{if}\ j \text{ is even and } j=n-1.
\end{cases}\\
\delta(j,b)=
\begin{cases}
j-1 &\text{if}\ j \text{ is even},\\
j+1 &\text{if}\ j \text{ is odd and } j\not=n-1,\\
j   &\text{if}\ j \text{ is odd and } j=n-1.
\end{cases}\\
\end{array}
$$

Fig.~\ref{reverse folding1} and~\ref{reverse folding2} show the automata $\mathrsfs{A}_8$ and $\mathrsfs{A}_9$, respectively.

\begin{figure}[th]
\begin{center}
\unitlength=1mm
\begin{picture}(134,40)(3,0)
\node(A)(112,20){$0$}
\node(B)(94,20){$1$}
\node(C)(79,5){$2$}
\node(D)(79,35){$3$}
\node(E)(64,20){$4$}
\node(F)(46,20){$5$}
\node(G)(28,20){$6$}
\node(H)(10,20){$7$}
\drawloop[loopangle=0](A){$a,b$}
\drawedge[curvedepth=0](B,A){$a$}
\drawedge[curvedepth=-3](B,D){$b$}
\drawedge[curvedepth=-3](C,B){$b$}
\drawedge[curvedepth=3](C,E){$a$}
\drawedge[curvedepth=4](D,C){$a$}
\drawedge[curvedepth=-4](D,C){$b$}
\drawedge[curvedepth=3](E,D){$a$}
\drawedge[curvedepth=3](E,F){$b$}
\drawedge[curvedepth=3](F,E){$b$}
\drawedge[curvedepth=3](F,G){$a$}
\drawedge[curvedepth=3](G,F){$a$}
\drawedge[curvedepth=3](G,H){$b$}
\drawedge[curvedepth=3](H,G){$b$}
\drawloop[loopangle=180](H){$a$}
\end{picture}
\caption{The automaton $\mathrsfs{A}_8$}
\label{reverse folding1}
\end{center}
\end{figure}
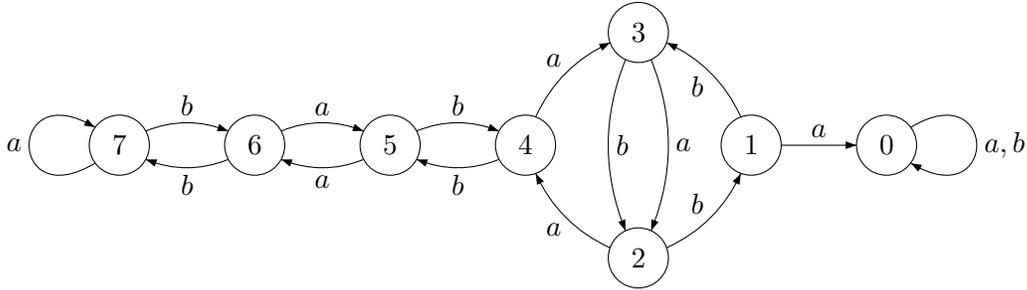

\begin{figure}[th]
\begin{center}
\unitlength=1mm
\begin{picture}(124,36)(3,0)
\node(A)(116,18){$0$}
\node(B)(100,18){$1$}
\node(C)(87,5){$2$}
\node(D)(87,31){$3$}
\node(E)(74,18){$4$}
\node(F)(58,18){$5$}
\node(G)(42,18){$6$}
\node(H)(26,18){$7$}
\node(I)(10,18){$8$}
\drawloop[loopangle=0](A){$a,b$}
\drawedge[curvedepth=0](B,A){$a$}
\drawedge[curvedepth=-3](B,D){$b$}
\drawedge[curvedepth=-3](C,B){$b$}
\drawedge[curvedepth=3](C,E){$a$}
\drawedge[curvedepth=4](D,C){$a$}
\drawedge[curvedepth=-4](D,C){$b$}
\drawedge[curvedepth=3](E,D){$a$}
\drawedge[curvedepth=3](E,F){$b$}
\drawedge[curvedepth=3](F,E){$b$}
\drawedge[curvedepth=3](F,G){$a$}
\drawedge[curvedepth=3](G,F){$a$}
\drawedge[curvedepth=3](G,H){$b$}
\drawedge[curvedepth=3](H,G){$b$}
\drawedge[curvedepth=3](H,I){$a$}
\drawedge[curvedepth=3](I,H){$a$}
\drawloop[loopangle=180](I){$b$}
\end{picture}
\caption{The automaton $\mathrsfs{A}_9$}
\label{reverse folding2}
\end{center}
\end{figure}
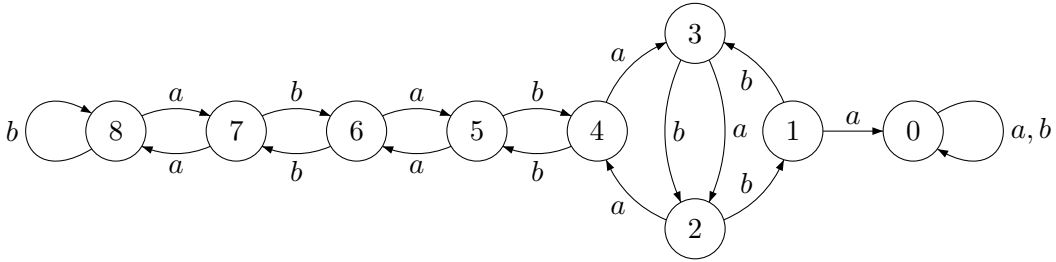

It is easy to see that $0$ is the sink state of the automaton $\mathrsfs{A}_n$, the letter $b$ acts as permutation on $Q_n$ and
the letter $a$ acts as permutation on $Q_n\setminus\{1\}$.  Thus, $\mathrsfs{A}_n$ is an almost permutation automaton for every $n$.

\begin{theorem}
For every $n\ge 5$, the reset threshold of the automaton $\mathrsfs{A}_n$ is at least $4n-13$.
\end{theorem}
\begin{proof}
Let $W$ be a reset word of minimal length for the automaton $\mathrsfs{A}_n$.

We say that an occurrence of the factor $aa$ of $W$ is {\it significant\/} if the letter of the word $W$ following this occurrence is $b$. We also say that an occurrence of the factor $bb$ of $W$ is {\it significant\/} if the letter of the word $W$ preceding this occurrence is $a$. It is easy to see that different significant factors of $W$ do not overlap. If a letter occurs in the word $W$ beyond any of its significant factors, we refer to this occurrence as {\it extra occurrence}.

Let $u$ be the longest prefix of the word $W$ with the property $\{5,6,\dots ,n-1\}\subseteq Q_n.u$. Let $w$ be such that $W=uw$.
Since the transition $\delta(5,b)=4$ is the only transition from the set $\{5,6,\dots ,n-1\}$ to the set $\{0,1,2,3,4\}$, the first letter of the word $w$ is $b$. Since $\delta(4,b)=5$, the second letter of the word $w$ is $a$ and $4\notin Q_n.u$.

Consider the actions of all words of length at most 5 to the set $Q_n$. We see that if $|u|<6$ then $Q_n.u=Q_n\setminus\{1,4\}$.
(More precisely, there are only 6 words $u$ of length at most 5 with the property $4\notin Q_n.u$, namely, $abaa$, $abba$, $babaa$, $babba$, $aabaa$, $aabba$. For every word $u$ from this list, we have $Q_n.u=Q_n\setminus\{1,4\}$.)

We call a state $q\in Q_n.u$ {\it essential\/} if there is a prefix $x$ of the word $w$ such that $\delta(q,x)=4$. As we noticed above, the word $w$ starts with the factor $ba$; therefore, if $3\in Q_n.u$, then $3$ is an essential state. Every state from $\{5,6,\dots ,n-1\}$ is also essential since the transition $\delta(5,b)=4$ is the only transition from the set $\{5,6,\dots ,n-1\}$.

Consider an essential state $q\in Q_n.u$. Let $x_q$ be the longest prefix of the word $w$ such that $\delta(q,x_q)=4$. If the word $v_q$ is such that $w=x_qv_q$, the length of $v_q$ is at least $4$ since $\delta(4,v_q)=\delta(q,w)=0$.

If the first letter of $v_q$ is $b$, then $\delta(q,x_qb)=5>4$, whence there is a prefix $y$ of $w$ which is longer than $x_q$ and $\delta(q,y)=4$. This contradicts the choice of the word $x_q$. Therefore the first letter of $v_q$ is $a$. Denote the second letter of $v_q$ by $\gamma$. If the third letter of $v_q$ is $a$ then $\delta(q,x_qa\gamma a)=4$. This again contradicts the choice of the word $x_q$.

Thus, $v_q$ starts with the factor $a\gamma b$, and it contains $a^2$ or $b^2$ as a significant factor. We refer to this significant factor as the significant factor {\it corresponding to the state\/} $q$.

Let $s,t\in Q_n.u$ be two different essential states. Suppose that $x_s=x_t$. Take the longest prefix $y$ of $x_s$ such that $\delta(s,y)\ne\delta(t,y)$ and denote the next letter that occurs in $x_s$ by $\tau$. Let $s_1=\delta(s,y)$ and $t_1=\delta(t,y)$. We have $s_1\ne t_1$ and $\delta(s_1,\tau)=\delta(t_1,\tau)$. This means that $\{s_1,t_1\}=\{0,1\}$, whence $\tau=a$ and $\delta(s_1,\tau)=\delta(t_1,\tau)=0$ by the definition of the automaton $\mathrsfs{A}_n$. Hence $\delta(s,x_s)=0$, and
this contradicts the choice of the word $x_s$. Thus, $x_s\ne x_t$. This means that the significant factors corresponding to the different states are different.

So, we have shown that either $|u|=4$ and $w$ contains at least $n-4$ significant factors corresponding to the states $3,5,6,\dots ,n-1$ or $|u|\geq 6$ and $w$ contains at least $n-5$ significant factors, corresponding to the states $5,6,\dots ,n-1$. Therefore the sum of the length of the word $u$ and the lengths of all significant factors in the word $w$ is at least $2n-4$.

Let $x$ be the shortest prefix of the word $w$ such that
$$Q_n.ux\cap \{5,6,\dots ,n-1\}=\varnothing$$ 
and let the word $y$ be such that $w=xy$. (We thus have $W=uxy$.) Since $\delta(5,b)=4$ is the only transition from the set $\{5,6,\dots ,n-1\}$, the definition of $x$ implies that $4\in Q.ux$, the last letter of the word $x$ is $b$ and the first letter of the word $y$ is $a$. Therefore the length of the suffix $y$ is at least $4$.

Since the length of $W$ is minimal, the last two letters of the word $W$ (and of the suffix $y$) are $ba$. Hence the last letter of the suffix  $y$ is an extra occurrence. Suppose that all other letters of $y$ occur within significant factors. Then $y=(a^2b^2)^ta$ for some $t$, but the equality $\delta(4,a^2b^2a)=2\ne0$ contradicts the choice of $W$ and the equality $\delta(4,a^2b^2a^2b)=5$ contradicts
the choice of $x$. Therefore, the suffix $y$ contains at least two extra occurrences.

Now we count extra occurrences in the word $x$. Let $x_1$ be the shortest prefix of the word $w$ such that $n-1\notin Q_n.ux_1$. This means that $n-2\in Q.ux_1$ and $x_1$ is a proper prefix of $x$. Let $x_2$ be such that $x=x_1x_2$. (We thus have $W=ux_1x_2y$.)

We notice that $p-1\leq \delta(p,\ell)\leq p+1$ for every state $p\in \{5,6,\dots ,n-1\}$ and every letter $\ell\in\Sigma$. Observe also that $X.a^2=X$ and $X.b^2=X$ for every $X\subseteq \{5,6,\dots ,n-1\}$. It means that the word $x_1$ contains at least $n-5$ extra occurrences and the word $x_2$ contains at least $n-6$ extra occurrences.

We conclude that the length of the word $W$ is at least
$$(2n-4)+2+(n-5)+(n-6)=4n-13.$$
This proves the theorem.
\end{proof}

It can be immediately verified that the words  $aba(abaabbab)^{\frac{n-5}{2}}aaba$ for each odd $n\geq 7$ and the words
$aba(abaabbab)^{\frac{n-8}{2}}ab^3(abaabbab)aaba$ for each even $n\geq 10$ are reset words for the corresponding automata
$\mathrsfs{A}_n$. This means that for $n=7,9,10,\dots $ the reset threshold of the automaton $\mathrsfs{A}_n$ is $4n-13$.
We have also verified that the reset threshold of the automaton $\mathrsfs{A}_8$ is $20$.

Observe that order of the letter $b$ in every automaton $\mathrsfs{A}_n$ is equal to $6$. Therefore, if we append to  $\mathrsfs{A}_n$ a tail of length $k=6\cdot\lceil\frac{n}{6}\rceil$, we obtain a binary \szan\ $\mathrsfs{B}_N$ with $N=n+k$ states and reset threshold
$$\frac{1}{4}N^2+2N+O(1).$$
In particular, if $n\equiv4\pmod6$, we can chose $k=n-4$ and get the desired series $\mathrsfs{B}_N$ with 
$$rt(\mathrsfs{B}_N)=\frac{1}{4}N^2+2 N-9.$$
We have thus proved

\begin{theorem}
For every $n\ge16$, $n\equiv4\pmod{12}$, there exists a binary \szan\ with $n$ states and reset threshold $\frac{1}{4}n^2+2n-9.$
\end{theorem}

We are sure that this new lower bound for reset threshold of binary \sza\ with $n$ states is not tight because our computational experiments have delivered further examples of almost permutation automata with reset threshold higher than $4n-13$ where $n$ is the number of states and these automata seem to extend to some infinite series. However, at the moment we cannot yet support this experimental evidence by rigorous proofs.

\end{document}